\documentclass{article}

\usepackage{arxiv}
\usepackage[utf8]{inputenc}
\usepackage{tikz,pgfplots}
\usepackage{graphicx}
\usetikzlibrary{shapes.geometric, arrows.meta, positioning}
\usetikzlibrary{shadows.blur}
\usetikzlibrary{decorations.shapes}
\usetikzlibrary{decorations.pathmorphing}
\usepgfplotslibrary{fillbetween}
\pgfplotsset{compat=newest}
\usepackage[
    natbib=true,
    style=numeric,
    sorting=none
]{biblatex}
\usepackage{biblatex}
\usepackage[english]{babel}
\usepackage{csquotes}
\usepackage[T1]{fontenc}    
\usepackage{hyperref}
\hypersetup{
    colorlinks=true,
    linkcolor=red,
    filecolor=magenta,      
    urlcolor=blue,
    pdftitle={Overleaf Example},
    pdfpagemode=FullScreen,
    }
\usepackage{authblk}
\usepackage{booktabs}       
\usepackage{amsfonts}       
\usepackage{nicefrac}       
\usepackage{microtype}      
\usepackage{lipsum}		
\usepackage{amsmath}

\usepackage{doi}
\usepackage{dcolumn}
\usepackage{bm}
\usepackage{amssymb}
\newtheorem{theorem}{Theorem}

\newtheorem{corollary}{Corollary}
\newtheorem{definition}{Definition}

\title{Further Results on Null and Force-Free Electromagnetic Fields}

\author[1,\thanks{gmenon@troy.edu}]{Govind Menon} 
\author[2,\thanks{radhikari@troy.edu}]{Rakshak Adhikari} 
\affil[1,2]{Center for Relativity and Cosmology\\ Troy University, Troy, AL 36082}

\begin{document}
\maketitle
\begin{abstract}
The theory of Force-Free Electrodynamics (FFE) provides a robust framework for modeling the magnetospheres of compact objects, where the electromagnetic field's energy density dominates the surrounding plasma. Central to this theory is the existence of two-dimensional integral submanifolds, or field sheets, which foliate the spacetime. While it is established that every null force-free field possesses an associated $2$-D null geodesic foliation (see \cite{Menon:2020hdk} for details), the converse—identifying which null geodesic congruences can support a force-free solution—remains a non-trivial computational challenge.
\vskip0.2in
In this paper, we extend the foliation-based approach to null FFE by addressing two primary obstacles to the existence of a solution: the equipartition of null mean curvature and the involutivity of the field sheet distribution. We prove a general existence theorem demonstrating that for any given null geodesic congruence, there always exists a local rotation of a $2$-D basis transverse to the geodesic congruence that satisfies the equipartition condition.  Furthermore, we establish that a shear-free null geodesic congruence is sufficient to guaranty the existence of an arbitrary function of three variables such that any choice of such a function will generate a null field sheet foliation. Additionally, each unique foliation will be associated with a null force-free field that further contains an arbitrary function of two variables.
\vskip0.2in
These results are formally linked to the vanishing of the shear tensor, providing a coordinate-independent geometric criterion for the existence of null FFE solutions. We illustrate these theorems with explicit examples in Schwarzschild and Kerr geometries and present new, non-trivial exact null solutions in flat spacetime and for the C-metric.

\end{abstract}


 \section{Introduction}
Force-free electrodynamics (FFE) provides the essential framework for modeling highly magnetized plasmas in extreme astrophysical environments, where the electromagnetic field energy density dominates over matter contributions. This regime is particularly relevant for understanding the magnetospheres of neutron stars, pulsars, and accreting black holes, where plasma inertia and pressure can be effectively neglected compared to electromagnetic stresses.

The governing equations of force-free electrodynamics were employed in the pioneering work of Goldreich and Julian \cite{Goldreich:1969sb} for pulsar magnetospheres, followed by the seminal paper of Blandford and Znajek \cite{BZ77}, which demonstrated that force-free plasma surrounding a rotating black hole can extract rotational energy via electromagnetic Poynting flux. 

Despite its foundational role in plasma astrophysics, the nonlinear structure of the force-free electrodynamics equations has long obstructed the systematic construction of exact solutions. Progress in elucidating magnetospheric dynamics has therefore largely depended on numerical approaches and perturbative schemes. Nonetheless, closed-form solutions remain indispensable, both as stringent tests of numerical methods and as vehicles for developing analytic understanding of the underlying physics and parameter dependence.

An important breakthrough in this direction came with the recasting of FFE within the framework of exterior calculus and differential geometry. The work of Gralla and Jacobson \cite{Gralla:2014yja} provided a fully covariant geometric framework that revealed the intrinsic mathematical structures underlying force-free fields. Building on this foundation, a foliation-based approach was developed \cite{Menon:2020hdk} that exploits a key geometric property: the kernel of the electromagnetic field tensor F forms an involutive distribution, naturally giving rise to a foliation of spacetime by two-dimensional submanifolds called field sheets. This geometric perspective has proven remarkably fruitful, yielding new exact solutions in both Kerr and FLRW spacetimes \cite{Adhikari:2023hhk, Adhikari:2024hre} as well as spacetimes where the metric is partially undetermined \cite{Adhikari:2025rya}. See also \cite{Compere:2016xwa} for a related foliation-based approach in the stationary, axisymmetric case.

The central result of the earlier foliation framework \cite{Menon:2020hdk} established that null force-free electromagnetic fields necessarily generate a special class of foliations characterized by an equipartition of null mean curvature termed null field sheet foliations. However, this result naturally raises the converse question: given a null geodesic congruence in a spacetime, under what conditions does there exist a corresponding null force-free electromagnetic field? This inverse problem is of both theoretical and practical importance, as it provides a systematic method for constructing FFE solutions directly from the geometric properties of null congruences.

The difficulty in addressing this converse problem lies in two distinct geometric requirements. First, for a given null geodesic congruence and an associated pair of orthogonal spacelike vectors, the equipartition condition—requiring that the null mean curvature distributes equally between the two orthogonal directions—may not be satisfied. Second, even when the equipartition condition holds, the resulting pair of vector fields may fail to form an involutive distribution, which is necessary for the existence of field sheet foliations.

In this paper, we systematically address both obstacles and provide complete answers. Our first main result demonstrates that the equipartition condition can always be satisfied: for any null geodesic congruence, there exist local orthonormal vector fields, obtained through an appropriate rotation, such that the equipartition of null mean curvature holds. This result removes the first barrier to constructing null FFE fields from arbitrary null congruences.

Our second main result addresses the involutivity requirement. We prove that when a null geodesic congruence admits uniform equipartition of null mean curvature, a slightly stronger condition, there exists a three-parameter family of null field sheet foliations, each associated with a two-parameter family of null force-free electromagnetic fields. The rotation angle that produces these involutive foliations is explicitly determined by an integration involving the commutator of the null and spacelike vectors.

A key geometric insight emerging from our analysis is the central role played by the shear tensor of the null congruence. We establish that the equipartition of null mean curvature is equivalent to the vanishing of the shear tensor when pulled back to the leaves of the foliation. Furthermore, uniform equipartition corresponds precisely to the complete vanishing of the shear tensor. This connection leads to an elegant corollary: any shear-free null geodesic congruence automatically admits a three-parameter family of null field sheet foliations and their associated force-free fields.

These results provide a systematic geometric framework for constructing null force-free solutions. Rather than searching for solutions to the nonlinear FFE equations directly, one can instead identify suitable null geodesic congruences with prescribed geometric properties and apply our constructive theorems. We illustrate the power of this approach through explicit examples in Schwarzschild spacetime, Kerr spacetime, and Minkowski space, demonstrating both the applicability of our theorems and the variety of solutions they generate.

The paper is organized as follows. In Section 2, we review the essential results from previous work on null foliations and force-free fields. Section 3 presents our main theorems on equipartition and involutivity, along with detailed examples in various spacetimes. We conclude by discussing the geometric significance of the shear tensor and its role in characterizing null force-free solutions.
 \section{Previous Results}

 This body of work extends the previous results on the theory of null and Force-Free Electrodynamics, henceforth FFE, developed in \cite{Menon:2020hdk}. There is a rich geometric structure underlying the theory of FFE. We begin with a quick survey of the aforementioned results as it pertains to null solutions.

In general relativity, the electromagnetic field is naturally described by a two-form $F$ satisfying
\begin{equation}
dF = 0 \;, \qquad *\, d * F = j \; ,
\end{equation}
where $d$ is the exterior derivative, $*$ denotes the Hodge star, and $j$ represents the current density. Force-free electrodynamics corresponds to the additional condition
\[
F(j^\sharp,\cdot)=0 \; .
\]

Here $\sharp$ denotes the musical isomorphism $T^*({\cal M}) \to T({\cal M})$, defined by
\[
j^\sharp(f)= g^{\mu\nu} j_\nu \,\partial_\mu(f)
\]
for any smooth function $f$ on the spacetime manifold ${\cal M}$, with $\flat$ denoting its inverse. At any point $p\in{\cal M}$, we introduce the scalar invariant
\[
F^2(p) := F_{\mu\nu}F^{\mu\nu}(p) \; .
\]
The electromagnetic field is magnetically dominated when $F^2(p)>0$, electrically dominated when $F^2(p)<0$, and null when $F^2(p)=0$. In recent years, the authors have constructed smooth solutions that exhibit transitions between these regimes. The remainder of this work is devoted to the study of null force-free configurations.

The kernel of $F$, denoted by $\ker F$, is a 2-dimensional subspace of the tangent bundle that satisfies the property that $i_v F =0$ whenever $v \in \ker F $. For $v, w \in \ker F$,
$$i_{[v,w]} F= [{\cal L}_v, i_w] F =0$$
since $F$ is a closed 2-form. Here, ${\cal L}_v$ denotes the Lie derivative with respect to the vector field $v$.
Therefore, $\ker F$ is an involutive distribution, meaning that whenever vector fields $v,w \in \ker F$, then $[v,w] \in \ker F$. Consequently, Frobenius' theorem implies that when a force-free $F$ exists on ${\cal M}$, spacetime can be foliated by 2-dimensional integral submanifolds of the distribution spanned by $\ker F$. The leaves of the foliation, which are the integral submanifolds of $\ker F$,  will be denoted as $ {\cal F}_a$. Here $a$ belongs to some indexing set $A$. The key points here are that
$$ {\cal F}_a \cap {\cal F}_b =0 \;{\rm whenever}\;a \neq b \in A\;,\;\;\;\cup_{a\in A} \;{\cal F}_a ={\cal M}\;,$$
and whenever $v\in T({\cal F}_a) $ for any $a \in A$ we have that $i_v F=0$. Following Gralla and Jacobson \cite{Gralla:2014yja}, ${\cal F}_a$ will be referred to as field sheets. It was shown in \cite{Menon:2020hdk} that for a null force-free field, $T({\cal F}_a)$ consists of a unique null geodesic congruence denoted by $l$. Furthermore, the metric when restricted to $T({\cal F}_a)$ is degenerate. I.e., $T({\cal F}_a)$ is locally spanned by a null geodesic vector field $l$ and a spacelike vector field $s$ that is orthogonal to $l$.

As shown in \cite{Menon:2020npo}, there exists a local adapted coordinate chart $(x^1, \dots, x^4)$ such that the slices given by constant values of 
$x^3$ and $x^4$ are indeed the field sheets. In such a chart

\begin{equation}
   F= f(x^3, x^4) \;dx^3 \wedge dx^4\;.
   \label{adaptedF}
\end{equation}

\vskip0.2in
\begin{definition}
    A foliation of ${\cal M}$ by $2$-dimensional submanifolds whose tangent space is spanned by a null geodesic vector field and an orthogonal spacelike vector field is referred to as a null foliation of ${\cal M}$ and is denoted by ${\cal F}_{2,N}$.
\end{definition}
\vskip0.2in
We complete the tetrad by picking additional vector fields $\alpha$ and $n$ such that 

\begin{equation}
  g(s,s)=g(\alpha, \alpha)=-g(n, l)=1, \;\;g(n,n)=
g(n, s)=g(n,\alpha)=g(\alpha, s)= g(\alpha, l) =0\;. 
\label{innerprod}
\end{equation}
\vskip0.2in

\begin{definition}
 We shall refer to the tetrad $(s, l, \alpha, n)$ with the above mentioned properties as a foliation-adapted frame for a null foliation ${\cal F}_{2,N}$.   
\end{definition}
\vskip0.2in
The null mean curvature, or the null expansion scalar, $\theta$, for the congruence generated by $l$ is given by
\begin{equation}
   \theta= \frac{1}{2} \;\Big[\;g(\nabla_s\; l, \;s)+g(\nabla_{\alpha} \;l, \;\alpha)\;\Big]\;. 
   \label{nmcdef}
\end{equation}
\vskip0.2in
\begin{definition}Let $(s, l, \alpha, n)$ be a foliation-adapted frame for a null foliation  ${\cal F}_{2,N}$. 
If
$$\theta= g(\nabla_s\; l, \;s)\;,$$
we say that ${\cal F}_{2,N}$ admits an equipartition of null mean curvature. This is equivalent to the requirement
$$g(\nabla_s\; l, \;s) = g(\nabla_\alpha\; l, \;\alpha)\;.$$
\label{equipartdef}
\end{definition}
\vskip0.2in
Clearly, from the arguments above, to every null force-free electrodynamic field, there is a corresponding null foliation of spacetime. The following theorem, which is one of the central results of \cite{Menon:2020hdk} provides the condition for the converse of the above statement. 
\vskip0.2in
 \begin{theorem}
Let ${\cal F}_{2,N}$ be a null foliation of a fixed, electrically neutral background spacetime ${\cal M}$, and let $(s, l, \alpha, n)$ be a null foliation adapted frame for ${\cal F}_{2,N}$. Then there exists a class of null force fields $F$ of the form
$$F=u \cdot \kappa \;\alpha^\flat \wedge l^\flat$$
if and only if ${\cal F}_{2,N}$ admits an equipartition of null mean curvature.
Here $\kappa$ is given by the expression
$$  \kappa=(\alpha_3 \;l_4-\alpha_4\; l_3)^{-1}\;,$$
where
\begin{equation}
  \left(
           \begin{array}{c}
             \alpha^\flat \\
             l^\flat \\
           \end{array}
         \right)=\left(
    \begin{array}{cc}
      \alpha_3 & \alpha_4 \\
      l_3 & l_4  \\
    \end{array}
  \right)\left(
           \begin{array}{c}
             dx^3 \\
             dx^4 \\
           \end{array}
         \right)\;.
         \label{chart2frame}
\end{equation}
Here $(x^1, \dots, x^4)$ is an adapted chart used in eq.(\ref{adaptedF}). Additionally, $u$ is any smooth function that is constant on the leaves of the foliation, i.e., $u = u(x^3,x^4)$.
\label{existunifianl}
\end{theorem}

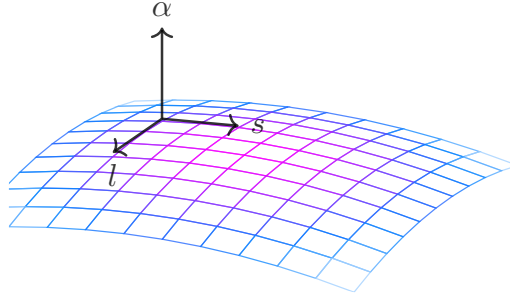
\begin{figure}
    \centering
\begin{tikzpicture}[scale=1.2]
\begin{axis}[xmin=-4, xmax=7,   ymin=-5,   ymax=7,zmin=-50,zmax=300 ,
   hide axis
    , colormap/cool,opacity=0.8];
\addplot3[
    mesh,samples=11,domain=-5:5,]{-((x)^2+(y)^2)};
\addlegendentry{$[l,s]\in \textrm{span}(l,s)$};


\draw [->,thick,] (-3,2,-10) -- (-3,2,100) node [above] {$\alpha$};

\draw [->,thick] (-3,2,-10) -- (-3,-1,-5) node [below] {$l$};
\draw [->,thick] (-3,2,-10) -- (-1,2,-3) node [right] {$s$};
\end{axis}
\end{tikzpicture}
    \caption{Null foliation generating a force-free field}
  
\end{figure}

\vskip0.2in
\begin{definition}
A null foliation ${\cal F}_{2,N}$ is a null field sheet foliation if it admits an equipartition of null mean curvature.
\end{definition}
\vskip0.2in
\begin{corollary}
 A null field sheet foliation will always permit a class of local null and force-free fields. Conversely, any null force-free solution will always have an associated null field sheet foliation.
\end{corollary}
\vskip0.2in
\begin{definition}
Let ${\cal F}_{2,N}$ be a null field sheet foliation of ${\cal M}$, and let $(s, l, \alpha, n)$ be an associated foliation-adapted frame. Then  $l$ admits a uniform equipartition of null mean curvature if
$$g(\nabla_s l, \alpha)+g(\nabla_\alpha l, s) =0\;. $$
\end{definition}

The requirement of uniform equipartition of null mean curvature allows us to widen the class of a null field foliation.

\begin{theorem}
Let ${\cal F}_{2,N}$ be a null field sheet foliation of ${\cal M}$, and let $(s, l, \alpha, n)$ be a null foliation adapted frame for ${\cal F}_{2,N}$. Let $l$  admit a uniform equipartition of null mean curvature, and for some smooth functions $A$ and $B$, let
$$ \hat s = A\; s+ B\;\alpha$$ be a unit vector field such that the span of $l$ and $\hat {s} $ forms an integrable distribution $\hat {\cal F}_{2,N}$ of ${\cal M}$. Then  $\hat {\cal F}_{2,N}$ is a null field sheet foliation.
\label{hatsoltheo}
\end{theorem}

\section{Further Results}
There are two major obstacles to a null geodesic congruence from being associated with a null and force-free electromagnetic field. First, the triplet $(l,s,\alpha)$ may not satisfy the equipartition condition in definition \ref{equipartdef}. The second is that once we find a triplet that satisfies the equipartition condition, the resulting pair $(l,s)$ may not form an involutive distribution.
In this section, for a given null geodesic congruence, we will look for ways to search for a particular set $(l,s)$ that is free of the problems listed above. The next, rather simple result, will show that given a null geodesic congruence, there will always be a triplet such that the equipartition condition is satisfied.
\subsection*{The Equipartition Condition}
\begin{theorem}
    Let $l$ be a null geodesic congruence. There are always local orthonormal vector fields $s$ and $\alpha$ both orthogonal to $l$ such that 
    $$g(\nabla_s\; l, \;s) - g(\nabla_\alpha\; l, \;\alpha)=0\;.$$
    \label{equigen}
\end{theorem}
{\bf Proof}. Before we begin the proof, it is important to note that we are not claiming the distribution spanned by $l$ and $s$ is involutive. 

Given $l$, pick any vector fields $s$ and $\alpha$ such that

$$g(l,l)=g(l,s)=g(l,\alpha) =g(s,\alpha)=0\;\;\;{\rm and}\;\;\;g(s,s)=g(\alpha, \alpha)=1\;.$$

Now, suppose the triplet does not satisfy the equipartition condition. Consider a rotation of type

\begin{equation}
  \left(
           \begin{array}{c}
             \hat s \\
             \hat \alpha \\
           \end{array}
         \right)=\left(
    \begin{array}{cc}
      \cos\Theta & \sin\Theta \\
      -\sin\Theta & \cos\Theta \\
    \end{array}
  \right)\left(
           \begin{array}{c}
             s \\
             \alpha \\
           \end{array}
         \right)\;,
\end{equation}

where $\Theta$ is a function on the manifold.
Then,

$$g(\nabla_{\hat s} l,\hat s)=\cos^2\Theta\; g(\nabla_s l,s)+\sin^2\Theta\; g(\nabla_\alpha l, \alpha)+\cos\Theta \sin\Theta \;\Big(g(\nabla_s l, \alpha)+g(\nabla_\alpha l,s)\Big)\;.$$
Similarly
$$g(\nabla_{\hat \alpha} l,\hat \alpha)=\cos^2\Theta\; g(\nabla_\alpha l,\alpha)+\sin^2\Theta\; g(\nabla_s l, s)-\cos\Theta \sin\Theta \;\Big(g(\nabla_s l, \alpha)+g(\nabla_\alpha l,s)\Big)\;.$$

Then
$g(\nabla_{\hat s}\; l, \;\hat s) = g(\nabla_{\hat \alpha}\; l, \;\hat \alpha)$
if and only if
\begin{equation}
  \tan(2\Theta) = \frac{g(\nabla_s\; l, \;s) - g(\nabla_\alpha\; l, \;\alpha)}{g(\nabla_s l, \alpha)+g(\nabla_{\alpha} l, s)}\;. 
  \label{rotangle}
\end{equation}

Please note that if $g(\nabla_s l, \alpha)+g(\nabla_{\alpha} l, s) =0$, a rotation by $\Theta = \pi/4$ will suffice. Clearly $\hat s$ and $\hat \alpha$ satisfy the requirements of the theorem.

\hfill $\blacksquare$
\subsection*{Example 1}
Consider,  the Schwarzschild line element written in the ingoing Eddington-Finkelstein Coordinates as follows \footnote{Throughout the paper we will be using the $(-,+,+,+)$ signature.}
 \[
ds^2 = \frac{-r + 2M}{r}\,dt^2 + 2\,dt\,dr + r^2\,d\theta^2 + r^2 \sin^2\theta\,d\varphi^2.
\]
Now, consider the following null geodesic congruence.\\

$$l=\frac{(c_2r^2 + \sqrt{c_2^2 r^3 + 2Mc_1^2 - c_1^2r}\sqrt r)}{r(-r + 2M)}\;\partial_{ t} -\frac{\sqrt {c_2^2r^3 + 2Mc_1^2 - c_1^2r}}{r^{3/2}}\;\partial_r +\frac{c_1}{r^2}\;\partial_\theta\;.$$
Let
$$s = \frac{1}{r\sin\theta}\partial_\varphi\;,$$
and
$$\alpha =\frac{c_1 (r-2M)}{ c_2r^2+\sqrt r \sqrt{c_2^2 r^3 + 2Mc_1^2 - c_1^2r}} \;\partial_r +\frac{1}{r}\ \partial_\theta\;.$$

It is easily verified that $s$ and $\alpha$ as given above satisfy the conditions of the theorem above.
Then

$$g(\nabla_sl ,s)-g(\nabla_\alpha l,\alpha) \neq 0\;.$$

A routine calculation will confirm that
$$g(\nabla_s l ,\alpha)+g(\nabla_\alpha l,s)=0\;.$$
Therefore, set $\Theta = \pi/4$, and so

$$\hat s = (s + \alpha)/\sqrt 2 \;\;\;{\rm and}\;\;\;
\hat \alpha = (s - \alpha)/\sqrt 2\;.$$
It is easily verified that now

$$g(\nabla_{\hat s} l ,\hat s)-g(\nabla_{\hat \alpha} l,\hat \alpha)=0\;.$$
\vskip0.2in \noindent 
Unfortunately, neither the pair $(l, \hat s)$ nor the pair$(l, \hat \alpha)$ is involutive unless $c_1 =0$, and this leads to the Schwarzschild limit of the well-known null solution in Kerr spacetime (see \cite{Menon:2005mg} and \cite{Brennan:2013kea} for its generalizations). 

\subsection*{Example 2}
In Boyer--Lindquist coordinates $(t,r,\theta,\varphi)$, the Kerr line element is
\[
\begin{aligned}
ds^2 ={} & -\left(1 - \frac{2Mr}{\rho^2}\right) dt^2
- \frac{4Mra\sin^2\theta}{\rho^2}\,dt\,d\varphi
+ \frac{\rho^2}{\Delta}\,dr^2
+ \rho^2\,d\theta^2   + \frac{\Sigma^2\sin^2\theta}{\rho^2}\,d\varphi^2 ,
\end{aligned}
\]
where
\[
\rho^2 = r^2 + a^2\cos^2\theta, \qquad
\Delta = r^2 - 2Mr + a^2, \qquad
\Sigma^2 = (r^2 + a^2)^2 - a^2\Delta \sin^2\theta .
\]

As a second non-trivial example, consider the following null geodesic congruence in this spacetime.
$$l=-\frac{f(\theta)}{\rho^2}\left(\frac{(\rho^2 (r^2+a^2)+4Ma^2 r\sin^2\theta)}{\Delta }\;\partial_t+\sqrt{(\rho^2)^2 +8Ma^2r\sin^2\theta} \;\partial_r-a\frac{(\rho^2-4Mr)}{\Delta }\;\partial_\varphi\right)$$
Let
$$s= \frac{1}{\sqrt{\rho^2}} \partial_\theta\;,$$
and
$$\alpha =\frac{1}{(\rho^2)^{3/2}}\Bigg(\frac{2aMr\sqrt{(\rho^2)^2+4Ma^2 r\sin^2\theta}}{\Delta}\;\partial_t-a\sin\theta (\rho^2-4Mr)\;\partial_r $$
$$- \frac{(\rho^2-2Mr)}{\Delta \sin\theta}\sqrt{(\rho^2)^2+4Ma^2 r\sin^2\theta}\;\partial_\varphi\Bigg)\;.$$
Once again, $s$ and $\alpha$ satisfy the required ortho-normality conditions. Here both 

$$g(\nabla_sl ,s)-g(\nabla_\alpha l,\alpha) \;\;\;{\rm and} \;\;\;g(\nabla_s l ,\alpha)+g(\nabla_\alpha l,s)$$ 

are non-zero. The rotation angle, as given by eq.(\ref{rotangle}) reduces to

$$\tan2\Theta=(a^2 \cos^2\theta - 3r^2)/(4ar\cos\theta)\;.$$

Lengthy but routine calculations confirm that the rotated vectors $\hat s$ and $\hat \alpha$ indeed satisfy

$$g(\nabla_{\hat s} l ,\hat s)-g(\nabla_{\hat \alpha} l,\hat \alpha)=0\;.$$

Sadly, neither $l,\hat s$ nor $l,\hat \alpha$ form an involutive pair and so they do not generate force-free solutions.

\subsection*{The Involutivity Condition}

Unlike the previous case of equipartition, the guaranteed existence for involutivity between $l$ and $\hat s$ requires the uniform equipartition of null mean curvature.

\begin{theorem}
Let $l$ be a null geodesic congruence, and let $(s, l, \alpha, n)$ be an associated tetrad with all inner products as listed in eq.(\ref{innerprod}).
Further, let 
$$g(\nabla_s\; l, \;s) = g(\nabla_\alpha\; l, \;\alpha)\;,\;\;\;{\rm and}\;\;\;g(\nabla_s l, \alpha)+g(\nabla_\alpha l, s) =0\;. $$

Then there exists an arbitrary function of three variables such that any choice of such a function will generate a null field sheet foliation. Additionally, each unique foliation will be associated with a null force-free field that further contains an arbitrary function of two variables.
\label{involthem}
\end{theorem}


{\bf Proof}. As per the conditions of the theorem, we assume that $l,s, \alpha$ satisfy the uniform equipartition of null mean curvature.
Using theorem \ref{hatsoltheo}, all that remains is to locate the involutive surfaces described by $A =\cos\Theta$ and $B =\sin\Theta$. As a first step, note that

$$g([l, \hat s],l)= g(\nabla_l \hat s, l)-g(\nabla_{\hat s}l , l)= -g(\hat s, \nabla_l l)-\hat s \;g(l,l)=0\;.$$

Therefore, all that remains to be shown for $(l,\hat s)$ to be an involutive pair is that
$g([l, \hat s],\hat \alpha) =0$. This is, in fact, a non-trivial requirement.

\begin{align}
    g([l, \hat s],\hat \alpha) =& -\sin\Theta\; l(\cos\Theta)-\sin\Theta \cos\Theta \;g([l,s],s)-\sin^2\Theta\; g([l, \alpha],s) +\cos^2\Theta \;g([l,s],\alpha)
  \nonumber\\
&
  +\cos\Theta\; l(\sin\Theta) +\sin\Theta \cos\Theta \;g([l,\alpha],\alpha)\;.
  \label{invol}
\end{align}

From the equipartition condition, we get that

$$g([l,\alpha],\alpha) = - g(\nabla_\alpha l, \alpha) = - g(\nabla_s l, s) = g([l,s],s)
\;.$$

Similarly, due to the uniform equipartition condition, we get that

$$g([l,s],\alpha) = g(\nabla_l s, \alpha)-g(\nabla_sl,\alpha) = -g(\nabla_l \alpha, s)+g(\nabla_\alpha l,s) = -g([l,\alpha],s)\;.$$

Putting the results of the last two equations into eq.(\ref{invol}), and setting $g([l, \hat s],\hat \alpha)=0$, we get that

$$\sin\Theta\; l(\cos\Theta) - \cos\Theta\; l(\sin\Theta) = g([l,s],\alpha)\;.$$
This gives the transport equation
\begin{equation}
    l(\Theta) = -\Gamma_{l}, 
    \qquad
    \Gamma_{l} := g([l, s], \alpha) .
\end{equation}

Here $\Gamma_{l}$ can be thought of as the precession rate of the transverse frame along  $l$.
The above equation is easily integrated by going to a local chart $\{x^\mu\}$ where $l = \partial/\partial x^1$ to obtain
\begin{equation}
 \Theta = -\int g([l,  s], \alpha)\; dx^1 + \tilde \Theta (x^2, x^3, x^4)\;.   
 \label{rotchoice}
\end{equation}
Here, $\tilde \Theta (x^2, x^3, x^4)$ is a $3$-parameter valued integration constant. Clearly, for each choice of an involutive foliation, theorems \ref{existunifianl} and \ref{hatsoltheo} give us a two-parameter valued null and force-free field.

\hfill $\blacksquare$

\subsection*{Example 3}
Here, we consider a sample null FFE solution in flat spacetime. 
In flat spacetime with Minkowski coordinates $(t,x,y,z)$, consider a null geodesic congruence generated by
$$l=\partial_t + \partial_z\;.$$
We choose two spacelike vector fields orthogonal to $l$ to be 
$$s = \partial_y,\;\;\;{\rm and}\;\;\;\alpha = \partial_x\;.$$
Clearly, the foliation generated by the distribution consisting of the span of $l$ and $s$ satisfies the uniform equipartition condition, and eq.(\ref{rotchoice}) reduces to
$$\Theta = \Theta(x,y,z-t)\;.$$

Hence, for any choice of $\Theta$, $l$ and $\hat s$ forms an involutive pair, where 
$$\hat s =
\cos\Theta\;\partial_x
+
\sin\Theta\;\partial_y.$$
Also, since
$$\hat \alpha = -\sin\Theta\,\partial_x +\cos\Theta \;\partial_y\;,$$

we can expect a null force-free solution of the form
$$F=u \cdot \kappa\; l^\flat \wedge \hat \alpha^\flat= u \cdot \kappa\;(-dt + dz) \wedge (-\sin\Theta \;dx+ \cos\Theta\; dy)\;.$$

For a general form of $\Theta(x,y,z-t)$, it is not easy to find an expression for $\kappa$; however, it is not difficult to now enforce the force-free condition on the simplified expression for $F$ given by

$$F = (-dt + dz) \wedge (A\;dx+ B dy)\;.$$
Then for an arbitrary function $B = B(x,y,z-t)$ of $3$ variables, we obtain the following expression for $A$ which has an additional $2$-parameter degree of freedom (as expected):

$$A = \int \partial_x B(x,y,z-t)\,dy + \tilde A(x,z-t)\;.$$

For a fixed $B$, the freedom in the choice of $A$ is contained in $\tilde A$.

\subsection*{The Role of the Shear Tensor in Null FFE}
For a null geodesic congruence $l$, the transverse space has the metric $h = s \otimes s + \alpha \otimes \alpha$. As usual, let
$$B_{\mu\nu} = h_\mu ^\lambda h_\nu ^\sigma \nabla_\lambda l_\sigma\;. $$ Then
$$B_{\mu\nu} = \frac{1}{2}\theta h_{\mu\nu} + s_\mu \alpha_\nu \;g(\nabla_\alpha l,s) +s_\nu \alpha_\mu \;g(\nabla_s l,\alpha)\;.$$
The antisymmetric part of $B_{\mu\nu}$ denoted by $\omega_{\mu\nu}$ is the twist of the congruence, and is given by
$$2 \omega = \Big(g(\nabla_\alpha l , s) - g(\nabla_s l , \alpha)\Big) (s \otimes \alpha - \alpha \otimes s)\;.$$
The shear tensor $\sigma$ is then given by the trace-free symmetric part defined by
$$\sigma_{\mu\nu} = B_{(\mu\nu)}- \frac{1}{2}\theta h_{\mu\nu}\;.$$
Here $()$ denotes symmetrization of the included indices. In terms of $s$ and $\alpha$ the shear tensor $\sigma$ corresponding to the congruence $l$ is now easily seen to be given by

\begin{equation}
2 \sigma = \Big(g(\nabla_s l , s) - g(\nabla_\alpha l , \alpha)\Big) (s \otimes s - \alpha \otimes \alpha)+\Big(g(\nabla_\alpha l , s) + g(\nabla_s l , \alpha)\Big)  (s \otimes \alpha + \alpha \otimes s)\;.
\label{sigmaexp}
\end{equation}

At last, we can give a geometric meaning to the conditions we have been using thus far.

\begin{corollary}
Let $(s, l, \alpha, n)$ be a foliation-adapted frame for a null foliation ${\cal F}_{2,N}$. Then the equipartition of null mean curvature is equivalent to the condition that the shear tensor $\sigma$ for the null congruence $l$ of ${\cal F}_{2,N}$ vanishes when it is pulled back to the leaves of the foliation in the sense that if $v_1, v_2 \in T( ({\cal F}_{2,N})_a)$, then $\sigma(v_1, v_2) =0$.
\end{corollary}
{\bf Proof}. This follows immediately from theorem \ref{existunifianl}.

\hfill $\blacksquare$

From eq.(\ref{sigmaexp}), $l$  admits a uniform equipartition of null mean curvature if and only if $\sigma$ vanishes.
Then, as an immediate consequence of theorems \ref{equigen} and \ref{involthem}, we obtain the following theorem.

\begin{theorem}
Let $l$ be a shear-free null geodesic congruence. Then, as before (theorem \ref{involthem}) , there exists an arbitrary function of three variables such that any choice of such a function will generate a null field sheet foliation. Additionally, each unique foliation will be associated with a null force-free field that further contains an arbitrary function of two variables.
\end{theorem}

All known exact null force-free solutions appear to be generated by a shear-free congruence (see \cite{Menon:2020hdk} for the well-known solution in Kerr, and an additional solution in FLRW spacetime). This begs the question as to whether all null force-free solutions are generated by a shear-free congruence? It appears not to be the case. The following theorem provides the necessary criterion.

\begin{theorem}
Let $(s, l, \alpha, n)$ be a foliation-adapted frame for a null foliation  ${\cal F}_{2,N}$. Then the null geodesic congruence $l$ associated to this foliation is shear-free if and only if $(l, \alpha)$ is also an involutive pair.
\end{theorem}
{\bf Proof}.
As in the proof of theorem \ref{involthem}, it is easy to see that $g([\alpha,l], l) =0$. It remains to look for the implications of $g([\alpha,l], s)$.

Since $(l,s)$ is an involutive pair, we have that
$$0=g([l,s], \alpha) = g(\nabla_l s, \alpha)-g(\nabla_s l, \alpha)\;.$$

Then 
$$g([\alpha,l], s)= g(\nabla_\alpha l,s)-g(\nabla_l \alpha,s)=g(\nabla_\alpha l,s) + g(\nabla_l s, \alpha) = g(\nabla_\alpha l,s)+g(\nabla_s l, \alpha)\;.$$

The last equality on the right-hand side comes from the equation above it.

\hfill $\blacksquare$

In section \ref{CMetric}, to illustrate the generality of our methods, we construct a null FFE solution in the C-metric. This was done with relative ease since we have a recipe whenever there is a shear-free null congruence.
\subsection*{Example 4}
As mentioned previously, the current catalog of null and force-free solutions has all been generated by a shear-free null geodesic. The following example clarifies the point that we do not need $l$ to be shear-free. 

Consider the following null geodesic congruence in Minkowski spacetime:

\begin{equation}
l=\partial_t+\sin\xi\; \partial_x+\cos\xi \; \partial_y\;.
\label{shearl}
\end{equation}
Here $\xi \equiv \xi(z)$ is an arbitrary function of $z$.  Set
\begin{equation}
 s=\tan\xi \;\partial_t + \sec \xi \; \partial_x\;,\;\;\;{\rm and}\;\;\;\;\alpha = \partial_z.   
 \label{st4}
\end{equation}

Then the equipartition and involutivity requirements hold, and we have that
\begin{align}
        F&=(u\cdot \kappa)\:l^\flat \wedge \alpha^\flat \nonumber\\
        &=f(\Phi,z)\:\left(dt \wedge dz-\sin{\xi}\:dx \wedge dz-\cos{\xi}\:dy\wedge dz\right)
        \label{sheary}
\end{align}
is a null and force-free solution. Here $f\equiv f(\Phi,z)$ where, $\Phi \equiv t-x\sin\xi(z)-y\cos\xi(z)$. The electric and magnetic fields are given by
\[
\mathbf{E} = -f(\Phi,z)\,\hat{\mathbf z},\qquad
\mathbf{B} = f(\Phi,z)\bigl(-\cos\xi(z)\,\hat{\mathbf x} + \sin\xi(z)\,\hat{\mathbf y}\bigr),\qquad
\]
and the Poynting vector is given by
\[\mathbf{S} = f^2(\Phi,z) \bigl(\sin\xi(z)\,\hat{\mathbf x} + \cos\xi(z)\,\hat{\mathbf y}\bigr).\]

This is a source-supported electromagnetic field whose energy flows horizontally in the xy-plane while the propagation direction rotates with height z. Here, for $l$, the shear tensor is given by
$$\sigma = \frac{1}{2}\;\frac{d\xi}{dz}\; (s \otimes \alpha + \alpha \otimes s)\;.$$
Since $l$ is not shear-free, no other rotation of $s$ and $\alpha$ will generate a further solution, and so we only have a $2$-parameter degree of freedom here; in contrast to the previous example.

This is the first solution, as far as the authors know, where we have a null FFE solution wherein the current density vector is not proportional to the null geodesic congruence (see section \ref{ShearyAPP}). This is also the first null FFE solution that we know where the null congruence is not shear-free. This connection is, in fact, a general one, as is shown in the following result.

\vskip0.2in
\begin{theorem}
    Let $F$ be a null force-free solution associated to the null geodesic congruence $l$. Then the current density vector $j$ is along $l$ iff $l$ is shear-free.
\end{theorem}
{\bf Proof}. In \cite{Menon:2020hdk}, it was shown that the current density vector is given by the expression

$$j=-(u \cdot \kappa)\;
  \left[\Big(\alpha(\ln u \cdot \kappa)+\;dl^\flat(n,\alpha)+ ds^\flat(\alpha, s)\Big)\; l+\;ds^\flat(l,\alpha)\;s\right]\;.$$
 On the other hand
 $$0=ds^\flat(l,\alpha)= l^\mu \alpha^\nu \left[\nabla_\mu s_\nu-\nabla_\nu s_\mu \right]=g(\nabla_l s, \alpha)-g(\nabla_\alpha s,l)=g(\nabla_s l, \alpha)+g(\nabla_\alpha l,s)\;.$$
 
 The last equality in the right-hand side above uses the involutivity of the kernel of $F$. Since $F$ is a solution, the equipartition condition is satisfied. The result then follows from eq.(\ref{sigmaexp}).
 
\hfill $\blacksquare$

We see from the above theorem that null and vacuum solutions must be generated by shear-free null congruences. This immediate result is one half of Robinson's theorem, see \cite{Robinson1961NullEM} \cite{Trautman1985SimpleRobinson} and \cite{Tafel1985}, which states that null and vacuum solutions of Maxwell's equations exist if and only if spacetime admits shear-free null geodesic congruences. The remainder of the Robinson theorem shows that there are choices for $u$ such that
$$\alpha(\ln u \cdot \kappa)+\;dl^\flat(n,\alpha)+ ds^\flat(\alpha, s)=0\;,$$ rendering $j=0.$
Furthermore, if $l$ has shear, we cannot get a vacuum solution. This is evident in Example 4, where $j$ is given in eq.(\ref{shearyj}). In particular, for a vacuum solution, we require 
that $\xi^\prime =0$, in which case $\sigma =0$.

\tikzset{
base/.style={
  rectangle,
  minimum height=1cm,
  draw=black,
  text width=5.6cm,
  align=center
},
startstop/.style={base, rounded corners, fill=green!20},
process/.style={base, fill=blue!20},
decision/.style={base, fill=yellow!20},
output/.style={base, rounded corners, fill=teal!25},
arrow/.style={
  thick,
  ->,
  shorten >=2pt,
  shorten <=2pt,
  rounded corners=3pt
}
}

\newcommand{\maxresult}{
Shear-free congruence\\
$\Rightarrow$ field sheet foliations generated by $\Theta(x^1,x^2,x^3)$\\
$\Rightarrow$ each foliation has solutions with $2$ degrees of freedom
}

\newcommand{\uniqueresult}{
Field sheets exist\\
$\Rightarrow$ each foliation has solutions with $2$ degrees of freedom
}

\newcommand{\noresult}{
No field sheets\\
$\Rightarrow$ no null force-free solution
}
\begin{figure}[htbp]
    \centering
    \begin{tikzpicture}[node distance=1.6cm, scale=0.85, every node/.style={transform shape}]
        \node (start) [startstop,fill=teal!25]
        {Given a null geodesic congruence $l$,\\
        $(\alpha,s)$ spacelike unit vectors orthogonal to $l$};

        \node (shear) [decision, below=of start]
        {Is $l$ shear-free?};

        \node (shearYes) [output, right=2cm of shear, fill=green!15]
        {\maxresult};

        \node (equip) [decision, below=of shear,fill=yellow!25]
        {Equipartition satisfied?\\
        $g(\nabla_s l,s)=g(\nabla_\alpha l,\alpha)$};
        \draw[arrow] (shear.south) -- node[left]{No} (equip.north);

        \node (rotateE) [process, below=of equip, fill=yellow!25]
        {Rotate $(s,\alpha)$ to enforce equipartition and obtain $(\hat s,\hat \alpha)$};

        \draw[arrow] (equip.south) -- node[left]{NO} (rotateE.north);
        \node (invol) [decision,  right=2.0cm of equip]
        {Is the pair $(l,s)$ or $(l,\alpha)$ involutive?};

        \node (involR) [decision, below=of rotateE]
        {Is the pair $(l,\hat{s})$ or $(l,\hat{\alpha})$ involutive?};

        \node (good1) [output, below=of invol,fill=green!16]
        {\uniqueresult};

        \node (good2) [output, below=of involR,fill=green!16]
        {\uniqueresult};

        \node (bad1) [output, right=2cm of invol,fill=red!25]
        {\noresult};

        \node (bad2) [output, right=2cm of involR,fill=red!25]
        {\noresult};
        \draw[arrow] (start) -- (shear);
        \draw[arrow]
          (shear.east)
          |- node[pos=0.25,above,xshift=25pt]{YES}
          (shearYes.west);
        \draw[arrow] (equip.east) -- node[above]{YES} (invol.west);
        \draw[arrow] (rotateE.south) -- (involR.north);
        \draw[arrow] (invol.south) -- node[left]{YES} (good1.north);
        \draw[arrow] (invol.east) -- node[midway, above, xshift=2mm]{NO} (bad1.west);
        \draw[arrow] (involR.south) -- node[left]{YES} (good2.north);
        \draw[arrow] (involR.east) -- node[midway, above, xshift=2mm]{NO} (bad2.west);
    \end{tikzpicture}
    \caption{Flowchart summarizing the geometric construction of null force-free electromagnetic fields from a prescribed null geodesic congruence. The procedure proceeds through shear check, frame rotation for equipartition, involutivity verification, and solution generation.}
    \label{fig:flowchart}
\end{figure}
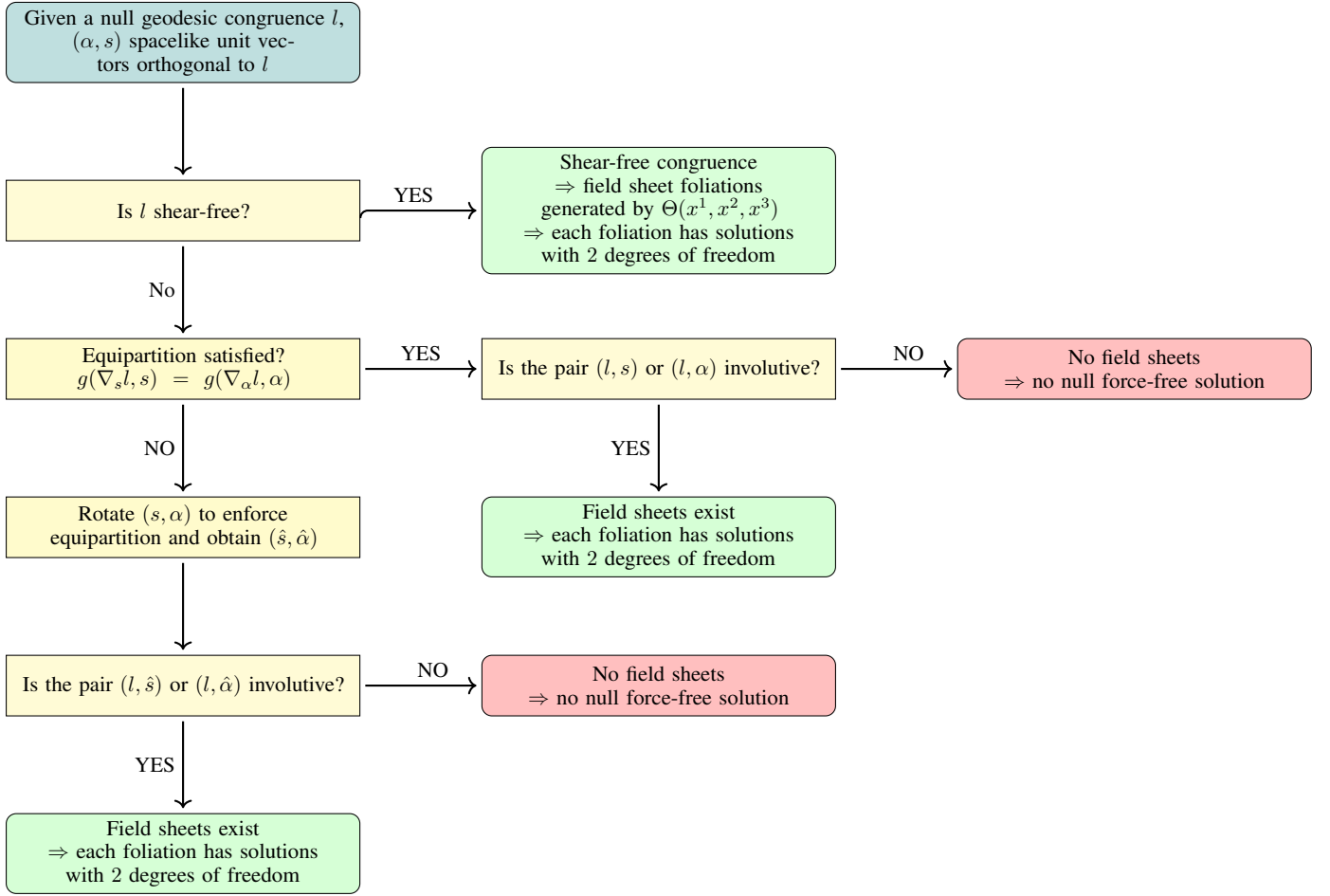

\section{Conclusion}
In this work, we developed a geometric framework that characterizes the existence of null force-free electromagnetic fields in terms of the properties of null geodesic congruences. By reversing the usual direction of analysis, the nonlinear force-free equations are replaced by transparent geometric requirements, equipartition of null mean curvature, and involutivity of the associated field sheet distribution, which can be addressed systematically at the level of congruence geometry rather than through direct solution of partial differential equations. We showed that equipartition is always achievable locally by an appropriate rotation of the transverse spacelike frame and therefore does not constitute an intrinsic obstruction to solution existence. The involutivity condition plays a more subtle role: uniform equipartition provides a sufficient condition for involutivity, guaranteeing the existence of an arbitrary function of three variables such that any choice of such a function will generate a null field sheet foliation. Additionally, each unique foliation will be associated with a null force-free field that further contains an arbitrary function of two variables. More generally, involutivity can persist even in the absence of uniform equipartition, though the resulting solution space is typically more constrained. A central outcome of our analysis is the identification of the shear tensor as the invariant governing this structure. Shear-free null geodesic congruences automatically admit maximal families of null force-free solutions, while nonzero shear restricts but does not eliminate the space of admissible fields and allows for configurations in which the current density is not aligned with the null direction.

The explicit examples presented in curved and flat spacetimes illustrate these possibilities, including a genuinely shearing null congruence in Minkowski space that nevertheless supports a null force-free solution. The construction procedure is summarized in the schematic flowchart of Figure~\ref{fig:flowchart}, which provides a practical algorithm for generating null force-free fields. Starting from a prescribed null geodesic congruence, the flowchart guides the reader through frame rotation to enforce equipartition, evaluation of involutivity, and systematic construction of null force-free fields whenever the geometric criteria are satisfied. Taken together, these results provide a systematic geometric characterization of the conditions under which null force-free electromagnetic fields arise from null geodesic congruences, and establish foliation geometry as a constructive tool for generating such solutions.
For readers who prefer the Newman--Penrose formalism, the Appendix reformulates these conditions in that language: equipartition is identified with the vanishing of \(\mathrm{Re}\,\sigma\), while uniform equipartition corresponds to the complete vanishing of the Newman--Penrose shear scalar \(\sigma\).
\printbibliography
\section{Appendix}
\subsection{Illustrative Features Contained in Example 4}
\label{ShearyAPP}
Considering the symmetry contained in the expression for $l$ in eq.(\ref{shearl}), it is natural to ask if we could have chosen 
$$\tilde s=\cot \xi\;\partial_t\;+\csc \xi\;\partial_y\;,\;\;\;{\rm and}\;\;\;\;\alpha = \partial_z$$
as a transverse basis in eq.(\ref{st4}). The answer is affirmative, since $s$ and $\tilde s$ only differ by a term proportional to $l$, which is a freedom contained in the theory of null geodesics. More precisely, 
$$\tilde s = -s + \frac{2}{\sin 2\xi} \;l\;.$$
Nonetheless, the solution given by eq.(\ref{sheary}) remains unchanged. This is so because the span of $(l,s)$ is the same as the span of $(l,\tilde s)$.
\vskip0.2in
It was noted earlier in example $3$ that, in general, and especially when the null congruence is shear-free, the expression for $\kappa$ is not easily obtained in a closed form. One reason for this is that we have a large class of foliations generated by rotations by angles $\Theta (x^2,x^3,x^4)$. However, when the null congruence has shear, there is only one possible choice of foliation. This is exactly the case in example $4$. So here we will briefly describe how we can obtain the functions $u$ and $\kappa$. Since $u$ is any smooth function that is constant on a field sheet, we must have that
$l (f)=0= s(f).$
This is satisfied by 
$f=f(\Phi ,z)$, and so we set set $u \equiv f$. It now remains to be shown that $\kappa =1$. It is easily seen that 

$$(x^1 = x, \;\;x^2 = y,\;\; x^3 =\Phi = t-x \sin \xi - y \cos \xi,\;\; x^4 =z)$$

is an adapted chart for our formalism. 
In the case of this example, eq.(\ref{chart2frame}) becomes
\begin{equation}
  \left(
           \begin{array}{c}
             \alpha^\flat \\
             l^\flat \\
           \end{array}
         \right)=\left(
    \begin{array}{cc}
      0 & 1\\
    -1 & (x \cos \xi-y \sin \xi ) \xi^\prime \\
    \end{array}
  \right)\left(
           \begin{array}{c}
             dx^3 \\
             dx^4 \\
           \end{array}
         \right)\;,   
\end{equation}
where $\prime$ denotes derivative with respect to $z$, and hence $\kappa =1$ as required. The charge current density vector field in this case is given by 

\begin{equation}
    j = f(x^3,z)  \Big( \zeta \; l - \xi'(z)\; \tilde s \Big)\;,
    \label{shearyj}
\end{equation}

$$
\zeta = \dfrac{
\xi'(z)\,(x \cos \xi(z) - y \sin \xi(z)) \, \dfrac{\partial f}{\partial x^3}(x^3, z)
- u(x^3, z) \, \tan \xi(z) \, y
- \dfrac{\partial f}{\partial z}(x^3, z)
}{f(x^3, z)}.
$$

Clearly, when the shear tensor of $l$ is non-trivial, $j$ is not along $l$; it also has a component along the  $\tilde s$ direction.
\vskip0.2in

\subsection{The Shear-Free Geodesic Congruence in the C-Metric and the Associated Null FFE Solution}
\label{CMetric}
As a concrete nontrivial example, we consider the C-metric in the Hong–Teo form introduced in \cite{Griffiths:2006tk}, which is a vacuum solution to the Einstein field equations and describes a pair of uniformly accelerating black holes. In this coordinate system, the spacetime admits a distinguished principal null direction that is geodesic and shear-free. 

Using this congruence, we demonstrate that the shear-free property alone is sufficient to support a family of null force-free electromagnetic fields. By our general construction, shear-free implies uniform equipartition and involutivity of the field-sheet distribution, ensuring the existence of a null force-free field generated by the geodesic congruence. This example highlights that the shear-free mechanism is not restricted to the familiar stationary geometries and flat space, but applies equally well to accelerating spacetimes.
Following \cite{Griffiths:2006tk}, the C-metric may be written in the Hong--Teo coordinates $(t,x,y,\varphi)$ as
$$\mathrm{d}s^2 = \frac{1}{A^2 (x+y)^2} \Big[ - H(y)\, \mathrm{d}t^2 + \frac{\mathrm{d}y^2}{H(y)} 
+ \frac{\mathrm{d}x^2}{G(x)} + G(x)\, \mathrm{d}\varphi^2\Big]\;,$$
where
$$G(x) = 1 - x^2 - 2 A M x^3\;,\;\;\;{\rm}\;\;\;H(y) = y^2 - 1 - 2 A M y^3\;,$$
and  $A, M$ are parameters of the geometry. The null, shear-free geodesic in this case is given by

$$l =(x + y)^2 \left( H(y)^{-1} \,\partial_t
  + \,\partial_y\right)\;.$$

As usual, define two spacelike vector fields orthogonal to $l$ by

$$s = (x+y)\, A \, \sqrt{G(x)} \;\partial_x\;,\;\;\;{\rm and}\;\;\;\alpha = - \frac{(x+y)\, A}{\sqrt{G(x)}} \;\partial_\varphi\;.$$

Set,
$$\beta = - \int H^{-1} dy+ t\;,$$
then $l(\beta)=0$. Notice that $\beta$ may not be globally defined because it is the integral of the inverse of a cubic function with possible roots. The null and force-free field corresponding to our configuration is given by

$$F = (x+y)\; \left( l^\flat \wedge (A(\beta,x,\varphi)\; \alpha^\flat + B(\beta,x,\varphi)\; s^\flat) \right)\;,$$
 where $A(\beta, x, \varphi)$ is an arbitrary function of its arguments (of three variables), and 
$$B(\beta,x,\varphi) = \int \!\Big[ (2 A M x^3 + x^2 - 1)\,\partial_\varphi A(\beta, x, \varphi)
      + (3 A M x^2 + x)\, A(\beta, x, \varphi) \Big] \,\mathrm{d}\varphi
      + \tilde B(\beta, x)\;.$$
$B(\beta,x,\varphi)$, as expected, contains an additional two degrees of freedom in $\tilde B(\beta, x)$.
Since the geodesic congruence is shear-free, the current is necessarily along $l$.

%
%

\subsection{Equipartition and Uniform Equipartition in the Newman-Penrose Dictionary}

Define  complex polarization vectors by
\begin{equation}
m := \frac{s + i\alpha}{\sqrt{2}}, 
\qquad 
\bar{m} := \frac{s - i\alpha}{\sqrt{2}}.
\end{equation}

We adopt the standard Newman--Penrose definitions
\begin{align}
\rho &:= -m^a \bar{m}^b \nabla_b \ell_a, \\
\sigma &:=- m^a m^b \nabla_b \ell_a,
\end{align}
so that \(\mathrm{Re}\,\rho\) is the expansion and \(\mathrm{Im}\,\rho\) is the twist. Expanding and collecting real and imaginary parts, the complex shear $\sigma$ is given by
\begin{align}
\sigma 
&= -\frac{1}{2} \left[ g(\nabla_s l, s) - g(\nabla_\alpha l, \alpha) \right]
   - \frac{i}{2} \left[ g(\nabla_s l, \alpha) + g(\nabla_\alpha l, s) \right].
\end{align}
I.e., the equipartition condition is given by 
$$\mathrm{Re}\,\sigma = 0\;,$$
and the uniform equipartition condition is the additional constraint
$$\mathrm{Im}\,\sigma = 0\;,$$
thereby rendering $\sigma$ trivial.

\end{document}